%% file: main_sbc.tex
\documentclass[12pt]{article}

\usepackage{sbc-template}

\usepackage{graphicx,url}
\usepackage{amsmath}
\usepackage{amsthm}
\usepackage{amssymb}
\usepackage{algorithm}
\usepackage{algpseudocode}
\usepackage{tikz}
\usetikzlibrary{arrows}
\usetikzlibrary{arrows.meta}
\usetikzlibrary{positioning, shapes.arrows, backgrounds}

\usepackage[utf8]{inputenc}  

\newtheorem{remark}{Remark}
\newtheorem{theorem}{Theorem}

\newtheorem{lemma}[theorem]{Lemma}

\newcommand\blfootnote[1]{%
  \begingroup
  \renewcommand\thefootnote{}\footnote{#1}%
  \addtocounter{footnote}{-1}%
  \endgroup
}

\title{Approximations for the Weighted Reversal, Transposition, and Indel Distance Problem with Intergenic Region Information}

\author{Gabriel Siqueira\inst{1}%\orcidID{0000-0001-5745-399X}
, Alexsandro Oliveira Alexandrino\inst{1}%\orcidID{0000-0002-6320-9747}
, Zanoni Dias\inst{1}%\orcidID{0000-0003-3333-6822}
}

\address{Instituto de Computação, Universidade Estadual de Campinas (Unicamp),\\ Campinas, Brazil\\
\email{\{gabriel.siqueira,alexsandro,zanoni\}@ic.unicamp.br}}

\begin{document} 

\maketitle

\begin{abstract}
Genome rearrangement distances are an established method in genome comparison. Works in this area may include various rearrangement operations representing large-scale mutations, gene orientation information, the number of nucleotides in intergenic regions, and weights reflecting the expected frequency of each operation. In this article, we model genomes containing at most one copy of each gene by considering gene sequences, with orientations, and representing intergenic regions according to their nucleotide lengths. We looked at a problem called Weighted Reversal, Transposition, and Indel Distance, which seeks the minimal cost sequence composed by the rearrangement operations of reversals, transposition, and indels, capable of transforming one genome into another. We leverage a structure called Labeled Intergenic Breakpoint Graph to show an algorithm for that problem with guaranteed approximations considering some sets of weights for the operations. \blfootnote{This work appeared in the Proceedings of the XVIII Brazilian Symposium on Bioinformatics (BSB'2025).} 
\end{abstract}

\section{Introduction}

In many scenarios of Computational Biology it is important to have some measure for the difference between two genomes considering some biological basis. If we are using genetic mutations, one such measure is the Genome Rearrangement distance, which seeks the shortest sequence of genome rearrangements (mutations that affect a large sequence of genetic material) capable of transforming one genome into the other. 

Two common rearrangements that are normally adopted in rearrangement problems are reversals, which invert a sequence of genetic material, and transpositions, which exchange two consecutive sequences of genetic material. Those two operations are called conservative, because they only change the order and orientation of sequences of the genetic material. There are also non-conservative operations, such as the indel, which inserts or removes a sequence of genetic material.

In this work, we assume that each gene can have at most one copy. With these restrictions and considering only the sequence of genes, each with an associated orientation, we have the following known results:
\begin{itemize}
    \item If we are only considering reversals, finding the rearrangement distance can be done in polynomial time~\cite{1999-hannenhalli-pevzner}.
    \item If we are considering both operations of reversals and transpositions, finding the rearrangement distance is NP-hard~\cite{2019b-oliveira-etal} and the best known approximation is $2$~\cite{1998-walter-etal}.
    \item If we are considering both operations of reversals and indels, finding the rearrangement distance can be done in polynomial time~\cite{2021-willing-etal}.
    \item If we are considering the operations of reversals, transpositions, and indels, finding the rearrangement distance is NP-hard and the best known approximation is $2$~\cite{2022a-alexandrino-etal,2023-alexandrino-etal}.
\end{itemize}

Beyond using the sequence of genes, some works in genome rearrangements also consider information presented in intergenic regions (normally represented by the number of nucleotides contained therein). The following results consider this additional information:
\begin{itemize}
    \item If we are only considering reversals, finding the rearrangement distance is NP-hard and the best known approximation is $2$~\cite{2021b-oliveira-etal}.
    \item If we are considering both operations of reversals and transpositions, finding the rearrangement distance is NP-hard and the best known approximation is $3$~\cite{2021a-oliveira-etal}.
    \item If we are considering both operations of reversals and indels, finding the rearrangement distance has unknown complexity and the best known approximation is $2.5$~\cite{2022b-alexandrino-etal}.
    \item If we are considering the operations of reversals, transpositions, and indels, finding the rearrangement distance is NP-hard and the best known approximation is $4$~\cite{2023b-alexandrino-etal}.
\end{itemize}

Another important variant of the genome rearrangement problems considers distinct weights for each operation. This is important to overcome a tendency the algorithms have to prefer transpositions, which may be a problem in populations where another rearrangement is more likely to occur~\cite{1996-blanchette-etal}. For instance, if reversals have a weight of $2$ and transpositions have a weight of $3$ there is an $1.5$-approximation algorithm for the rearrangement problem that does not consider intergenic regions~\cite{2019a-oliveira-etal} (if the ratio between the cost of reversals and transpositions is between 1 and 3 the approximation is lower or equal to 2 in this algorithm). In 2020, there was a study of weighted operations considering intergenic regions~\cite{2020a-brito-etal} with the description of approximations for different weight functions considering reversal, transpositions and indels, but the indels were only allowed to affect intergenic regions.

In this work, we take another step in the study of genome rearrangement problems with weighted operations by considering weights for the operations of reversal, transposition, and indel, such that the indels may also affect genes. For that goal, we apply some known results from the literature and prove the necessary facts to develop an algorithm and ensure its approximation factor.

In Section~\ref{sec:def}, we provide some necessary definitions and a formal description of the problem. Afterward, in Section~\ref{sec:bg}, we present the Labeled Intergenic Breakpoint Graph. Next, in Section~\ref{sec:algo} we describe an approximation algorithm for the rearrangement problem considering intergenic regions and different weights for the operations of reversal, transposition, and indel. Finally, Section~\ref{sec:conc} concludes the paper.

\section{Definitions}\label{sec:def}

A genome $\mathcal{G} = (S,\breve{S})$ is represented by a string $S$, with size $n$, and a list of integers $\breve{S}$, with size $n+1$. For $1 \le i \le n$, the element $S_i$, in the $i$-th position of $S$, represents the $i$-th gene of $\mathcal{S}$ and the $i$-th integer $\breve{S}_i$ of $\breve{S}$ represents the number of nucleotides  (size of the intergenic region) between $S_{i-1}$ and $S_{i}$, or before $S_1$ if $i=0$. There is also an integer $\breve{S}_{n+1}$ representing the size of the intergenic region at the end of the genome. The alphabet $\Sigma_S$ denotes the set of characters from $S$.  We also encode gene orientation by assigning a $+$ or $-$ sign to each element of $S$.

Given a genome $\mathcal{G} = (S,\breve{S})$, the \emph{reversal} $\rho^{(i,j)}_{(x,y)}$, with $1 \le i \le j \le n$, $0 \le x \le \breve{S}_{i}$, and $0 \le y \le \breve{S}_{j+1}$ is an operation that inverts a segment of $\mathcal{G}$. This segment starts in the $x$-th nucleotide of the intergenic region $\breve{S}_i$ and ends in the $y$-th nucleotide of the intergenic region $\breve{S}_{j+1}$. The order and orientation of the genes and the order of the intergenic regions on this segment are inverted, resulting in a new genome $\mathcal{G}_1 \cdot \rho^{(i,j)}_{(x,y)} = (S',\breve{S}')$. Next, we show the application of this operation:
\begin{align*}
    S &= (S_1~S_2~\ldots~\underline{S_i~S_{i+1}~\ldots~S_j}~\ldots S_{n-1}~S_n)\\
    \breve{S} &= (\breve{S}_1~\breve{S}_2~\ldots~(x\underline{+x')~\breve{S}_{i+1}~\breve{S}_{i+2}~\ldots~\breve{S}_j~(y+}y')~\ldots \breve{S}_{n}~\breve{S}_{n+1})\\
    S' &= (S_1~S_2~\ldots~\underline{-S_j~\ldots~-S_{i+1}~-S_i}~\ldots S_{n-1}~S_n)\\
    \breve{S} &= (\breve{S}_1~\breve{S}_2~\ldots~(x\underline{+y)~\breve{S}_{j}~\ldots~\breve{S}_{i+2}~\breve{S}_{i+1}~(x'+}y')~\ldots \breve{S}_{n}~\breve{S}_{n+1})\\
    x'&= \breve{S}_i - x,~~ y' = \breve{S}_{j+1} - y
\end{align*}

Given a genome $\mathcal{G} = (S,\breve{S})$, the \emph{transposition} $\tau^{(i,j,k)}_{(x,y,z)}$, with $1 \le i < j < k \le n+1$, $0 \le x \le \breve{S}_{i}$, $0 \le y \le \breve{S}_{j}$, and $0 \le z \le \breve{S}_{k}$, is an operation that exchanges two segments of $\mathcal{G}$. The first segment starts in the $x$-th nucleotide of the intergenic region $\breve{S}_i$ and ends in the $y$-th nucleotide of the intergenic region $\breve{S}_j$. The second segment starts in the $y$-th nucleotide of the intergenic region $\breve{S}_j$ and ends in the $z$-th nucleotide of the intergenic region $\breve{S}_k$. The genes and intergenic regions of these segments are swapped, resulting in a new genome $\mathcal{G} \cdot \tau^{(i,j,k)}_{(x,y,z)} = (S',\breve{S}')$. Next, we show the application of this operation:
\begin{align*}
    S &= (S_1~\ldots~\underline{S_i~\ldots~S_{j-1}}~\underline{S_j~\ldots~S_{k-1}}~\ldots~S_n)\\
    \breve{S} &= (\breve{S}_1~\ldots~(x\underline{+x')~\breve{S}_{i+1}~\ldots~\breve{S}_{j-1}~(y}~\underline{+y')\breve{S}_{j+1}~\ldots~\breve{S}_{k-1}~(z}+z')~\ldots~\breve{S}_{n+1})\\
    S' &= (S_1~\ldots~\underline{S_j~\ldots~S_{k-1}}~\underline{S_i~\ldots~S_{j-1}}~\ldots~S_n)\\
    \breve{S'} &= (\breve{S}_1~\ldots~(x\underline{+y')\breve{S}_{j+1}~\ldots~\breve{S}_{k-1}~(z}~\underline{+x')~\breve{S}_{i+1}~\ldots~\breve{S}_{j-1}~(y}+z')~\ldots~\breve{S}_{n+1})\\
    x'&= \breve{S}_i - x,~~ y' = \breve{S}_j - y,~~ z' = \breve{S}_k - z
\end{align*}

Given a genome $\mathcal{G} = (S,\breve{S})$, a sequence of genes (characters) $A$, with size $|A|$, and a sequence of intergenic regions (integers) $\breve{A}$, with size $|A| + 1$, the \emph{insertion} $\phi^{(i,A,\breve{A})}_{(x)}$, with $1 \le i \le n$ and $0 \le x \le \breve{S}_{i}$, is an operation that inserts the sequences $A$ and $\breve{A}$ in $\mathcal{G}$. The insertion occurs after the $x$-th nucleotide of the intergenic region $\breve{S}_i$. After the insertion, we have a new genome $\mathcal{G} \cdot \phi^{(i,A,\breve{A})}_{(x)}$. Note that, if $A$ is empty, the insertion only adds nucleotides in the intergenic region $\breve{S}_i$. Next, we show the application of this operation:
\begin{align*}
    S &= (S_1~S_2~\ldots~S_i~S_{i+1}~\ldots S_{n-1}~S_n)\\
    \breve{S} &= (\breve{S}_1~\breve{S}_2~\ldots~\breve{S}_{i-1}~(x+x')~\breve{S}_{i+1}~\ldots \breve{S}_{n}~\breve{S}_{n+1})\\
    S' &= (S_1~S_2~\ldots~S_i~\underline{A_1~A_2~\ldots~A_{|A|}}~S_{i+1}~\ldots S_{n-1}~S_n)\\
    \breve{S'} &= (\breve{S}_1~\breve{S}_2~\ldots~\breve{S}_{i-1}~(x+\underline{\breve{A}_1)~\breve{A_2}~\ldots~(\breve{A}_{|A|+1}+}x')~\breve{S}_{i+1}~\ldots \breve{S}_{n}~\breve{S}_{n+1})\\
    x'&= \breve{S}_i - x
\end{align*}

Given a genome $\mathcal{G} = (S,\breve{S})$, the \emph{deletion} $\psi^{(i,j)}_{(x,y)}$, with $1 \le i \le j \le n + 1$, $0 \le x \le \breve{S}_{i}$, and $0 \le y \le \breve{S}_{j}$, is an operation that removes a segment of $\mathcal{G}$. This segment starts in the $x$-th nucleotide of the intergenic region $\breve{S}_i$ and ends in the $y$-th nucleotide of the intergenic region $\breve{S}_{j}$. After the deletion, we have a new genome $\mathcal{G} \cdot \psi^{(i,j)}_{(x,y)}$. Note that, if $i=j$, the deletion only removes nucleotides in the intergenic region $\breve{S}_i$. In that case, we require that $0 \le x \le y \le \breve{S}_j$. Next, we show the application of this operation:
\begin{align*}
    S &= (S_1~S_2~\ldots~S_{i-1}~\underline{S_i~S_{i+1}~\ldots~S_{j-1}}~S_j~\ldots S_{n-1}~S_n)\\
    \breve{S} &= (\breve{S}_1~\breve{S}_2~\ldots~\breve{S}_{i-1}(x\underline{+x')~\breve{S}_{i+1}~\breve{S}_{i+2}~\ldots~\breve{S}_{j-1}~(y+}y')~\breve{S}_{j+1}~\ldots \breve{S}_{n}~\breve{S}_{n+1})\\
    S &= (S_1~S_2~\ldots~S_{i-1}~S_j~\ldots S_{n-1}~S_n)\\
    \breve{S} &= (\breve{S}_1~\breve{S}_2~\ldots~\breve{S}_{i-1}(x+y')~\breve{S}_{j+1}~\ldots \breve{S}_{n}~\breve{S}_{n+1})\\
    x'&= \breve{S}_i - x,~~ y' = \breve{S}_{j} - y
\end{align*}

When we want to discuss an operation that can be either an insertion or a deletion, we use the term indel. Given a genome $\mathcal{G}$ and a sequence of operations $R = (\beta_1,\beta_2,\ldots,\beta_{|R|})$, where each $\beta_i$ can be a reversal, a transposition, or an indel, we denote by $\mathcal{G} \cdot R$ the genome $\mathcal{G} \cdot \beta_1 \cdot \beta_2 \cdot \ldots \cdot \beta_{|R|}$, which is the result of applying in order each operation of R to $\mathcal{G}$.

We are interested in comparing two genomes $\mathcal{G}_1$ and $\mathcal{G}_2$ that may have a distinct set of genes. So we need to insert or delete characters from the genomes. If we are allowed to insert or remove any segment of the genome then the problem becomes trivial, just remove the whole genome $\mathcal{G}_1$ and insert the genome $\mathcal{G}_2$. Consequently, we assume that only genes that appear exclusively in $\mathcal{G}_1$ can be deleted and only genes that appear exclusively in $\mathcal{G}_2$ can be inserted. However, we do not restrict which intergenic regions can be inserted or deleted.

Considering that each gene can have at most one copy, we use the following representation for a pair of genomes $\mathcal{G}_1$ and $\mathcal{G}_2$. The genome $\mathcal{G}_2$, with $m$ genes, is represented by the identity string $\iota = (+1~+2~\ldots~+m)$, with all elements positive, and by the integer list $\breve{\iota}$. The genome $\mathcal{G}_1$, with $n$ genes, is represented by the string $S$ and integer list $\breve{S}$, such that each character of $S$ that corresponds to a gene present in $\mathcal{G}_2$ receives the same number used in $\iota$, and the other characters are represented by the symbol $\alpha$, because they need to be deleted and therefore do not need to be distinguished from one another. As we can delete or insert multiple consecutive characters and intergenic regions, we can combine any consecutive sequence of characters that appear only in $\iota$ or only in $S$ (the intergenic regions between those characters are removed from $\breve{\iota}$ and $\breve{S}$). Consequently, we assume that there are no two consecutive characters in $\iota$ that appear only in $\iota$, and there are no two consecutive characters $\alpha$ in $S$.

We want to attribute a weight to each operation (reversal, transposition, or indel). We denote the weight of a rearrangement $\beta$ by $W(\beta)$, and we denote the sum of the weights of all rearrangements of a sequence $R$ by $W(R)$. In the Intergenic Weighted Reversal, Transposition, and Indel Distance (IWRTID) problem, we are interested in a sequence of operations $R$ that turns $\mathcal{G}_1$ into $\mathcal{G}_2$ such that the value of $W(R)$ is minimized. This problem is NP-hard even if all operations have the same weight~\cite{2023b-alexandrino-etal}.

\section{Labeled Intergenic Breakpoint Graph}\label{sec:bg}

To develop the new algorithms, we will use the Labeled Intergenic Breakpoint Graph presented in a previous study~\cite{2021b-alexandrino-etal}. To define this graph, we include two characters $S_0 = \iota_0 = 0$ and $S_{n +1} = \iota_{m+1} = m+1$ at the beginning and the end of the genomes $\mathcal{G}_1$ and $\mathcal{G}_2$. Let $\pi^S$ be the string $S$ without the characters $\alpha$. We denote by $\Sigma_{\pi^S}$ the set of characters in $\pi^S$ and by $\Sigma_{\iota}$ the set of characters in $\iota$. We also defined $next(x,\Sigma_{\pi^S} \cap \Sigma_{\iota}) = min(y \in \Sigma_{\pi^S} \cap \Sigma_{\iota} | y > x)$, for $x \in \Sigma_{\pi^S} \setminus \{m+1\}$.

Given two genomes $\mathcal{G}_1 = (S,\breve{S})$ and $\mathcal{G}_2 = (\iota,\breve{\iota})$, the \emph{Labeled Intergenic Breakpoint Graph} $G(\mathcal{G}_1,\mathcal{G}_2) = (V,E,w,\ell)$ is a graph with the vertex set $V$, the edge set $E$, the weight function $w: E \to \mathbb{N}$, and the label function $\ell: E \to (\Sigma_{\iota} \setminus \Sigma_{\pi^S}) \cup \{\alpha\}$. The set $V$ has a vertex $+\pi^S_0$, a vertex $-\pi^S_{n+1}$, and, for each character $\pi^S_i, 1 \le i \le n$, two vertices $-\pi^S_i$ and $+\pi^S_i$. The edge set $E$ is separated into the set $E^o$ of origin edges and the set $E^t$ of target edges. For each $1 \le i \le n + 1$, there is an origin edge $o_i=(+\pi^S_{i-1},-\pi^S_i)$ with weight $w(o_i)$ equal to the sum of intergenic regions in $\breve{S}$ between the elements $\pi^S_{i-1}$ and $\pi^S_{i}$. The label $\ell(o_i)$ is equal to $\alpha$, if there is one $\alpha$ between $\pi^S_{i-1}$ and $\pi^S_{i}$, or is empty otherwise. For every $x \in \Sigma_{\pi^S} \setminus \{m+1\}$, there is a target edge $t_x = (+x,-next(x,\Sigma_{\pi^S} \cap \Sigma_{\iota}))$ with weight $w(t_x)$ equal to the sum of the intergenic regions in $\breve{\iota}$ between $x$ and $next(x,\Sigma_{\pi^S} \cap \Sigma_{\iota})$. The label $\ell(t_x)$ is equal to $x+1$, if $x+1 \ne next(x,\Sigma_{\pi^S} \cap \Sigma_{\iota})$, or is empty otherwise.

Each vertex from $V$ has only one incident target edge and one incident origin edge. Consequently, there is a unique decomposition of $G(\mathcal{G}_1,\mathcal{G}_2)$ into alternating cycles (cycles formed by interleave target and origin edges). One such cycle is called \emph{trivial} if it has only one target and one origin edge, and it is called \emph{non-trivial} otherwise. If some edge of a cycle has a label the cycle is called \emph{labeled}, and it is called \emph{clean} otherwise. A cycle $C$ is \emph{balanced} if the sum of the weights of the origin edges is equal to the sum of the weights of the target edges, and it is \emph{unbalanced} otherwise. An unbalanced cycle is \emph{positive} if the sum of the weights of the origin edges is lower than the sum of the weights of the target edges, and it is \emph{negative} otherwise. If a cycle is balanced and clean it is called \emph{good}, and it is called \emph{bad} otherwise (if it is unbalanced or labeled).

We use the following standard drawing of the graph $G(\mathcal{G}_1,\mathcal{G}_2)$, which will help in the following definitions. We place all vertices aligned in the order $+\pi^S_0$, $-\pi^S_1$, $+\pi^S_1$, $\ldots$, $-\pi^S_{n}$, $+\pi^S_{n}$, $-\pi^S_{n+1}$. The origin edges are shown as horizontal lines. Note that these edges always connect two neighbor vertices, considering the order in which they are drawn. The target edges are shown as arcs above the vertices. Figure~\ref{fig:bg} shows an example of the standard drawing of a breakpoint graph. In this graph, the black cycle is trivial, bad, labeled, and balanced, the red cycle is good, the blue cycle is bad, labeled, and positive, and the cyan cycle is bad, labeled, and negative.

\input{break_graph}

Considering the standard drawing, when traversing a cycle, starting from the rightmost vertices being traversed from right to left starting always using an origin edge, we may traverse an origin edge from right to left or from left to right. If all origin edges of cycles are traversed from right to left, we call this cycle \emph{convergent}, and we call it \emph{divergent} otherwise. Given a non-trivial convergent cycle $C$. If the edges are traversed in order from right to left, then the cycles $C$ is called \emph{non-oriented}, and it is called \emph{oriented} otherwise. In Figure~\ref{fig:bg}, the cyan cycle is divergent, the red cycle is convergent and oriented, and the blue cycle is convergent and non-oriented.

We will use $c(\mathcal{G}_1,\mathcal{G}_2)$ and $c_g(\mathcal{G}_1,\mathcal{G}_2)$ to denote the number of cycles and the number of good cycles in $G(\mathcal{G}_1,\mathcal{G}_2)$, respectively. For a rearrangement sequence $R$, let $\pi^S \cdot \beta$ be the resulting $\pi^S$ after the application of the sequence $R$ in $\mathcal{G}_1 = (S,\breve{S})$. We use $\Delta c(\mathcal{G}_1,\mathcal{G}_2,\beta) = (|\pi^S| + 1 - c(\mathcal{G}_1,\mathcal{G}_2)) - (|\pi^S \cdot R| + 1 - c(\mathcal{G}_1 \cdot R,\mathcal{G}_2))$ and $\Delta c_g(\mathcal{G}_1,\mathcal{G}_2,\beta) = (|\pi^S| + 1 - c_g(\mathcal{G}_1,\mathcal{G}_2)) - (|\pi^S \cdot R| + 1 - c_g(\mathcal{G}_1 \cdot R,\mathcal{G}_2))$ to describe the effect of this rearrangement on the number of cycles and vertices in the graph.

Note that any distinction in the order of the genes of $\mathcal{G}_1$ and $\mathcal{G}_2$ will produce non-trivial cycles in $G(\mathcal{G}_1,\mathcal{G}_2)$, any difference in the size of the intergenic regions in $\mathcal{G}_1$ and $\mathcal{G}_2$ will produce unbalanced cycles in $G(\mathcal{G}_1,\mathcal{G}_2)$, and any difference in the gene content of $\mathcal{G}_1$ and $\mathcal{G}_2$ will produce labeled cycles in $G(\mathcal{G}_1,\mathcal{G}_2)$. This leads us to the following remark.

\vspace{8pt}
\begin{remark}\label{remark:cyc}
$G(\mathcal{G}_1,\mathcal{G}_2)$ has only trivial good cycles if and only if $\mathcal{G}_1 = \mathcal{G}_2$.
\end{remark}

Note that any operation of reversal and transposition applied to $\mathcal{G}_1$ can be seen as an operation that affects the graph $G(\mathcal{G}_1,\mathcal{G}_2)$, changing the order of the vertices and replacing some origin edges according to the new order. The weights and labels of the new edges will be attributed according to the parameters of the operation and weights and labels of the edges that have been removed. Similarly, a deletion or insertion can be seen as changes to the weights and labels of the origin edges, with the possible inclusion of some vertices in the case of insertion. Note that an insertion that only affects intergenic regions corresponds to adding weight to one of the origin edges. Considering the convenience of this type of operation, we also define the operation of adding weight to one of the target edges with a clean label. This is a virtual insertion because there is no rearrangement operation directly correspondent to it, but the following lemma allows us to use it. Figure~\ref{fig:op} shows an example of the use of this lemma.

\vspace{8pt}
\begin{lemma}
If a sequence $R$ of rearrangement operations (including virtual insertions in target edges with clean label) can transform $G(\mathcal{G}_1,\mathcal{G}_2)$ into a graph with only good cycles, then there is a sequence $R'$ with at most as many rearrangements as $R$ that transforms $\mathcal{G}_1$ into $\mathcal{G}_2$.
\end{lemma}

\begin{proof}
After the application of the operations from the sequence $R$ in $G(\mathcal{G}_1,\mathcal{G}_2)$ we have a new graph $G(\mathcal{G}'_1,\mathcal{G}'_2)$, such that $\mathcal{G}'_1 = \mathcal{G}'_2$ (Remark~\ref{remark:cyc}). Note that $\mathcal{G}'_1$ is the result of applying all non-virtual operations in $\mathcal{G}_1$, because the virtual insertion only affects target edges. Besides, the virtual insertion is the only operation that affects target edges with clean labels. Consequently, the only distinction between $\mathcal{G}'_1$ and $\mathcal{G}_2$ is the size of some intergenic regions. We just have to apply deletions in each of these intergenic regions to turn $\mathcal{G}'_1$ into $\mathcal{G}_2$. We can take $R'$ to be the sequence of non-virtual operations of $R$ plus these deletions added to the end. Note that $R'$ has at most as many operations as $R$, because the number of distinct intergenic regions between $\mathcal{G}'_1$ and $\mathcal{G}_2$ is at most the number of virtual insertions in $R$.
\end{proof}

\input{op_sequence}

\section{Approximation Algorithms}\label{sec:algo}

In this section, we present the approximation algorithm for the IWRTID problem.  First, we establish how the operations can affect the values of $\Delta c$ and $\Delta c_g$.

\vspace{8pt}
\begin{lemma}\cite[Theorem 1]{1996-bafna-pevzner}\label{lemma:lb_r_c}
For any reversal $\rho$, we have $\Delta c(\mathcal{G}_1,\mathcal{G}_2,\rho) \le 1$.
\end{lemma}

\vspace{8pt}
\begin{lemma}\cite[Lemma 2]{2021b-alexandrino-etal}\label{lemma:lb_r_cg}
For any reversal $\rho$, we have $\Delta c_g(\mathcal{G}_1,\mathcal{G}_2,\rho) \le 1$.
\end{lemma}

\vspace{8pt}
\begin{lemma}\label{lemma:lb_d_c}
For any indel $\delta$, we have $\Delta c(\mathcal{G}_1,\mathcal{G}_2,\delta) \le 0$.
\end{lemma}

\begin{proof}
We can only delete characters $\alpha$, so deletions cannot have $\Delta c(\mathcal{G}_1,\mathcal{G}_2,\delta) \neq 0$. Besides, for each cycle inserted in $G(\mathcal{G}_1,\mathcal{G}_2)$ we must also insert at least one gene. Consequently, any change in $c(\mathcal{G}_1,\mathcal{G}_2))$ will be offset by a change in $|\pi^S|$, so insertions cannot have $\Delta c(\mathcal{G}_1,\mathcal{G}_2,\delta) > 0$.
\end{proof}

\vspace{8pt}
\begin{lemma}\cite[Lemma 3]{2021b-alexandrino-etal}\label{lemma:lb_d_cg}
For any indel $\delta$, we have $\Delta c_g(\mathcal{G}_1,\mathcal{G}_2,\delta) \le 1$.
\end{lemma}

\vspace{8pt}
\begin{lemma}\cite[Lemma 2.1]{1998-bafna-pevzner}\label{lemma:lb_t_c}
For any transposition $\tau$, we have $\Delta c(\mathcal{G}_1,\mathcal{G}_2,\tau) \le 2$.
\end{lemma}

\vspace{8pt}
\begin{lemma}\cite[Lemma 2]{2023b-alexandrino-etal}\label{lemma:lb_t_cg}
For any transposition $\tau$, we have $\Delta c_g(\mathcal{G}_1,\mathcal{G}_2,\tau) \le 2$.
\end{lemma}

In the next lemmas, we describe operations that can produce a positive value for $\Delta c$ or $\Delta c_g$ depending on the characteristics of a given cycle $C$ of the graph $G(\mathcal{G}_1,\mathcal{G}_2)$.

\vspace{8pt}
\begin{lemma}\cite[Lemma 5]{2021b-alexandrino-etal}\label{lemma:triv_1d}
If there is in $G(\mathcal{G}_1,\mathcal{G}_2)$ a trivial cycle $C$ with an empty origin edge, such that $C$ is unbalanced and the desire edge is empty or $C$ is non-negative and the target edge has a non-empty label, then there exists an indel $\delta$, such that $\Delta c(\mathcal{G}_1,\mathcal{G}_2,(\delta)) = 0$ and $\Delta c_g(\mathcal{G}_1,\mathcal{G}_2,(\delta)) = 1$.
\end{lemma}

\vspace{8pt}
\begin{lemma}\cite[Lemma 6]{2021b-alexandrino-etal}\label{lemma:triv_2d}
If there is in $G(\mathcal{G}_1,\mathcal{G}_2)$ a trivial bad cycle $C$, then there exists a sequence $R$ with up to $2$ indels, such that $\Delta c(\mathcal{G}_1,\mathcal{G}_2,R) = 0$ and $\Delta c_g(\mathcal{G}_1,\mathcal{G}_2,R) = 1$.
\end{lemma}

\vspace{8pt}
\begin{lemma}\label{lemma:unbal}
If there is in $G(\mathcal{G}_1,\mathcal{G}_2)$ an unbalanced clean cycle $C$, then there exists an indel $\delta$, such that $\Delta c(\mathcal{G}_1,\mathcal{G}_2,(\delta)) = 0$ and $\Delta c_g(\mathcal{G}_1,\mathcal{G}_2,(\delta)) = 1$.
\end{lemma}

\begin{proof}
If $C$ is positive, we apply an insertion that increases the weight of an origin edge of $C$ to make $C$ balanced. If $C$ is negative, we apply a virtual insertion (that will be turned into a deletion in the end) that increases the weight of a target edge of $C$ to turn $C$ balanced.
\end{proof}

\vspace{8pt}
\begin{lemma}\cite[Lemma 4.6]{2021a-oliveira-etal}\label{lemma:good_ori_3t}
If there is in $G(\mathcal{G}_1,\mathcal{G}_2)$ an good oriented cycle $C$, then there exists a sequence $R$ with $3$ transpositions, such that $\Delta c(\mathcal{G}_1,\mathcal{G}_2,R) = 2$ and $\Delta c_g(\mathcal{G}_1,\mathcal{G}_2,R) = 2$.
\end{lemma}

\vspace{8pt}
\begin{lemma}\label{lemma:good_ori_1t_2d}
If there is in $G(\mathcal{G}_1,\mathcal{G}_2)$ an good oriented cycle $C$, then there exists a sequence $R$ with $1$ transposition and up to $2$ indels, such that $\Delta c(\mathcal{G}_1,\mathcal{G}_2,R) = 2$ and $\Delta c_g(\mathcal{G}_1,\mathcal{G}_2,R) = 2$.
\end{lemma}

\begin{proof}
By Lemma 4.3 from a previous work~\cite{2021a-oliveira-etal} we can apply one transposition that splits $C$ into three cycles such that one of them is balanced, and therefore good, because $C$ is clean. We can balance the other two cycles with one indel each, as stated in Lemma~\ref{lemma:unbal}.
\end{proof}

\vspace{8pt}
\begin{lemma}\label{lemma:bad_ori_1t_1d}
If there is in $G(\mathcal{G}_1,\mathcal{G}_2)$ a bad oriented cycle $C$, then there exists a sequence $R$ with uma transposition and up to um indel, such that $\Delta c(\mathcal{G}_1,\mathcal{G}_2,R) = 2$ and $\Delta c_g(\mathcal{G}_1,\mathcal{G}_2,R) = 1$.
\end{lemma}

\begin{proof}
By Lemma 12 from a previous work~\cite{2023b-alexandrino-etal}, there is a transposition that turns $C$ into three cycles such that one of them is a trivial non-negative cycle with a clean source edge. If necessary, we can apply Lemma~\ref{lemma:triv_1d} to turn this cycle good with one indel.
\end{proof}

\vspace{8pt}
\begin{lemma}\cite[Lemma 7]{2021b-alexandrino-etal}\label{lemma:bad_div_1r_1d}
If there is in $G(\mathcal{G}_1,\mathcal{G}_2)$ a labeled divergent cycle $C$, then there exists a sequence $R$ with $1$ reversal and up to $1$ indel, such that $\Delta c(\mathcal{G}_1,\mathcal{G}_2,R) = 1$ and $\Delta c_g(\mathcal{G}_1,\mathcal{G}_2,R) = 1$.
\end{lemma}

\vspace{8pt}
\begin{lemma}\cite[Lemma 8]{2021b-alexandrino-etal}\label{lemma:good_div_1r}
If there is in $G(\mathcal{G}_1,\mathcal{G}_2)$ a good divergent cycle $C$, then there exists one reversal, such that $\Delta c(\mathcal{G}_1,\mathcal{G}_2,R) = 1$ and $\Delta c_g(\mathcal{G}_1,\mathcal{G}_2,R) = 1$.
\end{lemma}

\vspace{8pt}
\begin{lemma}\label{lemma:bad_nodiv_3r_2d}
If there is no divergent cycles in $G(\mathcal{G}_1,\mathcal{G}_2)$ and there are still nontrivial or bad cycles, then there exists a sequence $R$ with $3$ reversals and up to $2$ indels, such that $\Delta c(\mathcal{G}_1,\mathcal{G}_2,R) = 1$ and $\Delta c_g(\mathcal{G}_1,\mathcal{G}_2,R) = 2$.
\end{lemma}

\begin{proof}
If $C$ is oriented, the result is direct from Lemma 9 from a previous work~\cite{2022b-alexandrino-etal}. If $C$ is non-oriented, the result is direct from Lemma 10 from the same work~\cite{2022b-alexandrino-etal}.
\end{proof}

\begin{algorithm}
\caption{The input is a pair of genomes $\mathcal{G}_1$ and $\mathcal{G}_2$. The output is a sequence $S''$ of reversal, transposition and indel operations such that $\mathcal{G}_1
  \cdot S'' = \mathcal{G}_2$.}\label{algorithm:RTI_APX}
\begin{algorithmic}
\State $G \gets G(\mathcal{G}_1,\mathcal{G}_2)$
\State $S \gets [~]$
\While{$G$ has a cycle that is not good and trivial.}
    \If{there is a trivial bad cycle in $G$.}
        \State $S' \gets [\delta]$ or $[\delta_1, \delta_2]$ \Comment{Step I: lemmas~\ref{lemma:triv_1d} or~\ref{lemma:triv_2d}.}
    \ElsIf{there is an unbalanced clean cycle in $G$.}
        \State $S' \gets [\delta]$ \Comment{Step II: Lemma~\ref{lemma:unbal}}
    \ElsIf{there is a good oriented cycle in $G$.}
        \State $S' \gets [\tau_1,\tau_2,\tau_3] \text{ or } [\tau,\delta_1,\delta_2]$ \Comment{Step III: lemmas~\ref{lemma:good_ori_3t} and~\ref{lemma:good_ori_1t_2d}}
    \ElsIf{there is a bad oriented cycle in $G$.}
        \State $S' \gets [\tau,\delta]$ \Comment{Step IV: Lemma~\ref{lemma:bad_ori_1t_1d}.}
    \ElsIf{there is a labeled divergent cycle in $G$}
        \State $S' \gets [\rho,\delta]$ \Comment{Step V: Lemma~\ref{lemma:bad_div_1r_1d}.}
    \ElsIf{there is a good divergent cycle in $G$}
        \State $S' \gets [\rho]$ \Comment{Step VI: Lemma~\ref{lemma:good_div_1r}.}
    \Else
        \State $S' \gets [\rho_1,\rho_2,\rho_3,\delta_1,\delta_2]$ \Comment{Step VII: Lemma~\ref{lemma:bad_nodiv_3r_2d}.}
    \EndIf
    \State $G \gets G \cdot S'$
    \State $S \gets S + S'$
\EndWhile \\
$S'' \gets S$ with the virtual insertions replaced by deletions at the end of the sequence.
\Return $S''$
\end{algorithmic}
\end{algorithm}

Let us consider the weight of an operation as a cost to apply it. To prove the approximation of Algorithm~\ref{algorithm:RTI_APX}, we use the following value to describe the variation in $\Delta c_g$ and $\Delta c$ per weight of the operations applied.
$$
    \Delta c c_g(p_1,p_2,\mathcal{G}_1,\mathcal{G}_2,R) = \frac{p_1 \Delta c_g(\mathcal{G}_1,\mathcal{G}_2,R) + p_2 \Delta c(\mathcal{G}_1,\mathcal{G}_2,R)}{W(R)}
$$

In this equation for $\Delta c c_g$, the parameters $p_1$ and $p_2$ control the contribution of $\Delta c_g$ and $\Delta c$ and are useful to adjust this equation depending on the weights we are considering for the operations.

In the following theorem, we show approximations for some choices of weights. The weights are motivated by the fact that transpositions cut the genome in more points, which makes them preferred by most algorithms.
To investigate how weight variations affect approximation performance, we also tested scenarios with equal weights for reversals and indels, as well as cases where reversals have higher weights compared to indels.

\vspace{8pt}
\begin{theorem}
Algorithm~\ref{algorithm:RTI_APX} ensures the following approximations for the IWRTID problem:
\begin{itemize}
    \item $3.33$, if $W(\rho) = 2$, $W(\tau) = 3$, $W(\delta) = 2$, $p_1 = 4$, and $p_2 = 1$.
    \item $2.67$, if $W(\rho) = 2$, $W(\tau) = 3$, $W(\delta) = 1$, $p_1 = 1$, and $p_2 = 1$.
    \item $2.5$, if $W(\rho) = 1$, $W(\tau) = 2$, $W(\delta) = 1$, $p_1 = 4$, and $p_2 = 1$.
    \item $2$, if $W(\rho) = 2$, $W(\tau) = 4$, $W(\delta) = 1$, $p_1 = 1$, and $p_2 = 1$.
\end{itemize}
\end{theorem}

\begin{proof}
Considering Remark~\ref{remark:cyc} and the fact that at each iteration the algorithm finds a rearrangement sequence with positive value for $\Delta c$ or $\Delta c_g$, eventually, the algorithm stops producing the desired sequence $S''$, capable of transforming $\mathcal{G}_1$ into $\mathcal{G}_2$. We must prove that the weight of this sequence is at most the desire approximation factor times the optimal weight. For that goal, it is enough to prove that this fact is valid for each step.

For a reversal $\rho$, we know that $\Delta c c_g(p_1,p_2,\mathcal{G}_1,\mathcal{G}_2,(\rho)) \le \frac{p1+p2}{W(\rho)}$ (lemmas~\ref{lemma:lb_r_c} and~\ref{lemma:lb_r_cg}). Similarly, for a transposition $\tau$, $\Delta c c_g(p_1,p_2,\mathcal{G}_1,\mathcal{G}_2,(\tau)) \le \frac{2(p1+p2)}{W(\tau)}$ (lemmas~\ref{lemma:lb_t_c} and~\ref{lemma:lb_t_cg}), and for an indel $\delta$, $\Delta c c_g(p_1,p_2,\mathcal{G}_1,\mathcal{G}_2,(\delta)) \le \frac{p1}{W(\delta)}$ (lemmas~\ref{lemma:lb_d_c} and~\ref{lemma:lb_d_cg}). Consider $\Delta_{max} = max  \left (\frac{p1+p2}{W(\rho)}, \frac{2(p1+p2)}{W(\tau)}, \frac{p1}{W(\delta)} \right )$, which is the maximum possible value of $\Delta c c_g$ produce by a single operation. Note that if a sequence of operations reaches this value, then the sequence is optimal.

Let us look at the value of $\Delta cc_g(p_1,p_2,\mathcal{G}_1,\mathcal{G}_2,R)$ for each step of Algorithm~\ref{algorithm:RTI_APX}:
\begin{itemize}
    \item Step I: $\Delta_I=\Delta cc_g(p_1,p_2,\mathcal{G}_1,\mathcal{G}_2,R) = \frac{p_1}{2W(\delta)}$
    \item Step I: $\Delta_I=\Delta cc_g(p_1,p_2,\mathcal{G}_1,\mathcal{G}_2,R) = \frac{p_1}{W(\delta)}$
    \item Step III: $\Delta_{II}=\Delta cc_g(p_1,p_2,\mathcal{G}_1,\mathcal{G}_2,R) = \frac{2p_1 + 2p_2}{min(3W(\tau),W(\tau) + 2W(\delta))}$
    \item Step IV: $\Delta_{III}=\Delta cc_g(p_1,p_2,\mathcal{G}_1,\mathcal{G}_2,R) = \frac{p_1 + 2p_2}{W(\tau) + W(\delta)}$
    \item Step V: $\Delta_{IV}=\Delta cc_g(p_1,p_2,\mathcal{G}_1,\mathcal{G}_2,R) = \frac{p_1 + p_2}{W(\rho) + W(\delta)}$
    \item Step VI: $\Delta_{IV}=\Delta cc_g(p_1,p_2,\mathcal{G}_1,\mathcal{G}_2,R) = \frac{p_1 + p_2}{W(\rho)}$
    \item Step VII: $\Delta_V=\Delta cc_g(p_1,p_2,\mathcal{G}_1,\mathcal{G}_2,R) = \frac{2p_1 + 2p_2}{3W(\rho) + 2W(\delta)}$
\end{itemize}

For each step $K$, the applied sequence is within a factor of $\frac{\Delta_{max}}{\Delta_k}$ from the optimal. Consequently, the approximation factor of algorithm~\ref{algorithm:RTI_APX} is $max( \frac{\Delta_{max}}{\Delta_I}$, $\frac{\Delta_{max}}{\Delta_{II}}$, $\frac{\Delta_{max}}{\Delta_{III}}$, $\frac{\Delta_{max}}{\Delta_{IV}}$, $\frac{\Delta_{max}}{\Delta_{V}}$, $\frac{\Delta_{max}}{\Delta_{VI}},\frac{\Delta_{max}}{\Delta_{VII}})$. By replacing the values for $W(\rho)$, $W(\tau)$, $W(\delta)$, $p_1$, and $p_2$, we reach the informed approximations.
\end{proof}

\section{Conclusion}\label{sec:conc}

In this work, we investigated the genome rearrangement problem considering intergenic regions represented by their lengths and allowing reversal, transposition, and indel operations. We developed approximation algorithms with factors between $2$ and $3.33$ for distinct weights attributed to each operation.

Our contribution represents a meaningful step toward developing algorithms for rearrangement distances that consider operation weights. The main novelty of our approach is the integration of intergenic regions and indels in such problems. Nevertheless, as an initial effort, our study remains restricted to single-copy genes and a simplified representation of intergenic regions.

In future works, this problem can be explored with some alternative representation of intergenic regions or without the restriction of a single copy of each gene. It is also relevant to perform tests with real genomes to observe if the weights are helpful in the estimation of the evolutionary distance between them.

\section*{Acknowledgment} This work was supported by the S\~ao Paulo Research Foundation, FAPESP (grant
2021/13824-8% 06/2022-02/2025 (Doutorado Gabriel)
).
\bibliographystyle{sbc}
\bibliography{bibfile}

\end{document}

%% file: break_graph.tex
\begin{figure}[tb]
\centering
\resizebox{\textwidth}{!}{
\begin{tikzpicture}[scale=0.6]
\scriptsize

\begin{scope}[every node/.style={inner sep=1.5pt, minimum size = 0pt}]
    \node[circle, draw] (p0) at (0,0) {$~+0$};
    \node[circle, draw] (m1) at (1.5,0) {$~-1$};
    \node[circle, draw] (p1) at (3,0) {$~+1$};
    \node[circle, draw] (p8) at (4.5,0) {$~+8$};
    \node[circle, draw] (m8) at (6,0) {$~-8$};
    \node[circle, draw] (m4) at (7.5,0) {$~-4$};
    \node[circle, draw] (p4) at (9,0) {$~+4$};
    \node[circle, draw] (m2) at (10.5,0) {$~-2$};
    \node[circle, draw] (p2) at (12.0,0) {$~+2$};
    \node[circle, draw] (p7) at (13.5,0) {$~+7$};
    \node[circle, draw] (m7) at (15.0,0) {$~-7$};
    \node[circle, draw] (p9) at (16.5,0) {$~+9$};
    \node[circle, draw] (m9) at (18.0,0) {$~-9$};
    \node[circle, draw] (m5) at (19.5,0) {$~-5$};
    \node[circle, draw] (p5) at (21.0,0) {$~+5$};
    \node[circle, draw] (m10) at (22.5,0) {$-10$};

\end{scope}

\begin{scope}[>={Stealth[black]},
              every edge/.style={draw=black}]
    \path [-] (p0) edge [black] node [black, pos=0.5, sloped, above, yshift=0.00cm] {$3$} (m1);
    \path [-] (p0) edge [black] node [black, pos=0.5, sloped, below, yshift=0.00cm] {\color{gray} $\alpha$} (m1);
    \path [-] (p1) edge [red] node [black, pos=0.5, sloped, above, yshift=0.00cm] {$2$} (p8);
    \path [-] (m8) edge [blue] node [black, pos=0.5, sloped, above, yshift=0.00cm] {$0$} (m4);
    \path [-] (p4) edge [red] node [black, pos=0.5, sloped, above, yshift=0.00cm] {$1$} (m2);
    \path [-] (p2) edge [blue] node [black, pos=0.5, sloped, above, yshift=0.00cm] {$2$} (p7);
    \path [-] (m7) edge [cyan] node [black, pos=0.5, sloped, above, yshift=0.00cm] {$2$} (p9);
    \path [-] (m9) edge [red] node [black, pos=0.5, sloped, above, yshift=0.00cm] {$3$} (m5);
    \path [-] (p5) edge [cyan] node [black, pos=0.5, sloped, above, yshift=0.00cm] {$2$} (m10);
\end{scope}

\begin{scope}[>={Stealth[black]},
              every edge/.style={draw=black}]
    \path [-] (p0) edge  [black,bend left=55, dashed] node [black, pos=0.5, above] {$3$} (m1);
    \path [-] (p1) edge [red,bend left=80, dashed] node [black, pos=0.5, above] {$0$} (m2);
    \path [-] (p2) edge  [blue,bend right=55, dashed] node [black, pos=0.5, above] {$2$} (m4);
    \path [-] (p2) edge  [blue,bend right=55, dashed] node [black, pos=0.5, below] {\color{gray} $3$} (m4);
    \path [-] (p4) edge  [red,bend left=55, dashed] node [black, pos=0.5, above] {$3$} (m5);
    \path [-] (p5) edge  [cyan,bend right=55, dashed] node [black, pos=0.5, above] {$2$} (m7);
    \path [-] (p5) edge  [cyan,bend right=55, dashed] node [black, pos=0.5, below] {\color{gray} $6$} (m7);
    \path [-] (p7) edge  [blue,bend right=55, dashed] node [black, pos=0.5, above] {$1$} (m8);
    \path [-] (p8) edge  [red,bend left=55, dashed] node [black, pos=0.5, above] {$3$} (m9);
    \path [-] (p9) edge  [cyan,bend left=55, dashed] node [black, pos=0.5, above] {$1$} (m10);
\end{scope}

\end{tikzpicture}
}
\caption{Standard drawing of the graph $G(\mathcal{G}_1,\mathcal{G}_2)$ created from two genomes $\mathcal{G}_1 = (({+0}~{+\alpha}~{+1}~{-8}~{+4}~{+2}~{-7}~{-9}~{+5}~{+10}), ({1}~{2}~{2}~{0}~{1}~{2}~{2}~{3}~{2}))$ and $\mathcal{G}_2 = (({+0}~{+1}~{+2}~{+3}~{+4}~{+5}~{+6}~{+7}~{+8}~{+9}~{+10}), ({3}~{0}~{1}~{1}~{3}~{0}~{2}~{1}~{3}~{1}))$. Each color represents a different cycle of the cycle decomposition of $G(\mathcal{G}_1,\mathcal{G}_2)$. The weights (representing the intergenic regions) are indicated above each edge, and the labels (encoding the genes exclusive to one of the genomes) are indicated below each edge.}
\label{fig:bg}
\end{figure}
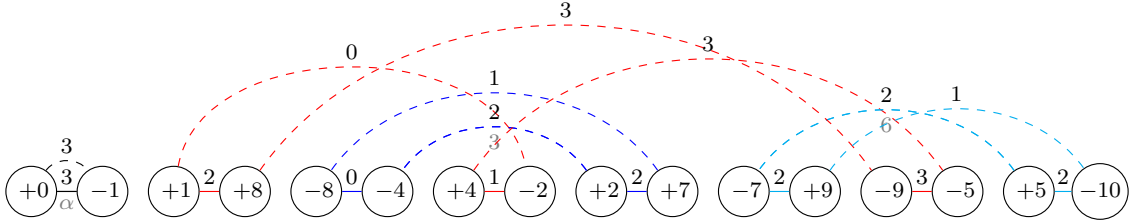

%% file: op_sequence.tex
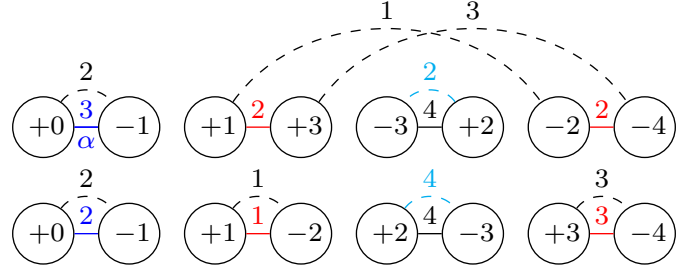
\begin{figure}[tb]
\centering
\resizebox{0.6\textwidth}{!}{
\begin{tikzpicture}[scale=0.6]
\scriptsize

\begin{scope}[every node/.style={inner sep=1.5pt, minimum size = 0pt}]
    \node[circle, draw] (p0) at (0,0) {$~+0$};
    \node[circle, draw] (m1) at (1.5,0) {$~-1$};
    \node[circle, draw] (p1) at (3,0) {$~+1$};
    \node[circle, draw] (p3) at (4.5,0) {$~+3$};
    \node[circle, draw] (m3) at (6,0) {$~-3$};
    \node[circle, draw] (p2) at (7.5,0) {$~+2$};
    \node[circle, draw] (m2) at (9.0,0) {$~-2$};
    \node[circle, draw] (m4) at (10.5,0) {$~-4$};

\end{scope}

\begin{scope}[>={Stealth[black]},
              every edge/.style={draw=black}]
    \path [-] (p0) edge [blue] node [blue, pos=0.5, sloped, above, yshift=0.00cm] {$3$} (m1);
    \path [-] (p0) edge [blue] node [blue, pos=0.5, sloped, below, yshift=0.00cm] {$\alpha$} (m1);
    \path [-] (p1) edge [red] node [red, pos=0.5, sloped, above, yshift=0.00cm] {$2$} (p3);
    \path [-] (m3) edge [black] node [black, pos=0.5, sloped, above, yshift=0.00cm] {$4$} (p2);
    \path [-] (m2) edge [red] node [red, pos=0.5, sloped, above, yshift=0.00cm] {$2$} (m4);
\end{scope}

\begin{scope}[>={Stealth[black]},
              every edge/.style={draw=black}]
    \path [-] (p0) edge  [black,bend left=55, dashed] node [black, pos=0.5, above] {$2$} (m1);
    \path [-] (p1) edge [black,bend left=55, dashed] node [black, pos=0.5, above] {$1$} (m2);
    \path [-] (p2) edge  [cyan,bend right=55, dashed] node [cyan, pos=0.5, above] {$2$} (m3);
    \path [-] (p3) edge  [black,bend left=55, dashed] node [black, pos=0.5, above] {$3$} (m4);
\end{scope}

\end{tikzpicture}
}

\resizebox{0.6\textwidth}{!}{
\begin{tikzpicture}[scale=0.6]
\scriptsize

\begin{scope}[every node/.style={inner sep=1.5pt, minimum size = 0pt}]
    \node[circle, draw] (p0) at (0,0) {$~+0$};
    \node[circle, draw] (m1) at (1.5,0) {$~-1$};
    \node[circle, draw] (p1) at (3,0) {$~+1$};
    \node[circle, draw] (m2) at (4.5,0) {$~-2$};
    \node[circle, draw] (p2) at (6,0) {$~+2$};
    \node[circle, draw] (m3) at (7.5,0) {$~-3$};
    \node[circle, draw] (p3) at (9,0) {$~+3$};
    \node[circle, draw] (m4) at (10.5,0) {$~-4$};

\end{scope}

\begin{scope}[>={Stealth[black]},
              every edge/.style={draw=black}]
    \path [-] (p0) edge [blue] node [blue, pos=0.5, sloped, above, yshift=0.00cm] {$2$} (m1);
    \path [-] (p1) edge [red] node [red, pos=0.5, sloped, above, yshift=0.00cm] {$1$} (m2);
    \path [-] (m3) edge [black] node [black, pos=0.5, sloped, above, yshift=0.00cm] {$4$} (p2);
    \path [-] (p3) edge [red] node [red, pos=0.5, sloped, above, yshift=0.00cm] {$3$} (m4);
\end{scope}

\begin{scope}[>={Stealth[black]},
              every edge/.style={draw=black}]
    \path [-] (p0) edge  [black,bend left=55, dashed] node [black, pos=0.5, above] {$2$} (m1);
    \path [-] (p1) edge [black,bend left=55, dashed] node [black, pos=0.5, above] {$1$} (m2);
    \path [-] (p2) edge  [cyan,bend left=55, dashed] node [cyan, pos=0.5, above] {$4$} (m3);
    \path [-] (p3) edge  [black,bend left=55, dashed] node [black, pos=0.5, above] {$3$} (m4);
\end{scope}

\end{tikzpicture}
}

\caption{Given $\mathcal{G}_1 = (({+0}~{+\alpha}~{+1}~{-3}~{-2}~{+4}), ({1}~{2}~{2}~{4}~{2}))$ and $\mathcal{G}_2 = (({+0}~{+1}~{+2}~{+3}~{+4}), ({2}~{1}~{2}~{3}))$. A sequence of rearrangement operations being applied to the graph $G(\mathcal{G}_1,\mathcal{G}_2)$, and the correspondent sequence applied to the genome $\mathcal{G}_1$ producing the genome $\mathcal{G}_2$. The sequence is composed of: a virtual insertion in the cyan edge correspondent to the deletion $\psi^{(4,4)}_{(0,2)}$; a reversal affecting the red edges correspondent to the reversal $\rho^{(4,5)}_{(1,0)}$; a deletion in the blue edge correspondent to the deletion $\psi^{(1,2)}_{(0,0)}$.}
\label{fig:op}
\end{figure}